\documentclass[twocolumn,superscriptaddress,preprintnumbers,showpacs,tightenlines,notitlepage]{revtex4}
\usepackage[latin1]{inputenc}
\usepackage{graphicx}
\usepackage{amssymb}
\usepackage{color}
\usepackage{float}
\usepackage{amsmath}
\usepackage{amsfonts}
\usepackage{dcolumn}
\usepackage{hyperref}
\usepackage{amsthm}
\usepackage{enumerate}
\usepackage{bm}
\usepackage{latexsym}
\usepackage{epsfig}
\usepackage{graphicx,psfrag,subfigure}
\usepackage{rotating}
\usepackage{tabularx}
\usepackage{epstopdf}
\newtheorem{prop}{Proposition}

\newtheorem{defn}{Definition}
\newtheorem{lem}{Lemma}
\newtheorem{thm}{Theorem}

 \newcommand{\be}{\begin{equation}} \newcommand{\ee}{\end{equation}}
\newcommand{\bea}{\begin{eqnarray}} \newcommand{\eea}{\end{eqnarray}}
\newcommand{\bse}{\begin{subequations}} \newcommand{\ese}{\end{subequations}}

\begin{document}
\title{Gravitational collapse of generalised Vaidya spacetime} 
\author{Maombi D. Mkenyeleye} 
\email{mkenyeleye@yahoo.co.uk\\Permanent Address:School of Mathematical Sciences,\\ University of Dodoma, Tanzania.}
\affiliation{Astrophysics and Cosmology Research Unit, School of Mathematics, Statistics and Computer Science, University of KwaZulu-Natal, Private Bag X54001, Durban 4000, South Africa.}
\author{Rituparno Goswami}
\email{Goswami@ukzn.ac.za}
\affiliation{Astrophysics and Cosmology Research Unit, School of Mathematics, Statistics and Computer Science, University of KwaZulu-Natal, Private Bag X54001, Durban 4000, South Africa.}
\author{Sunil D. Maharaj}
\email{Maharaj@ukzn.ac.za}
\affiliation{Astrophysics and Cosmology Research Unit, School of Mathematics, Statistics and Computer Science, University of KwaZulu-Natal, Private Bag X54001, Durban 4000, South Africa.}

\begin{abstract}
We  study the gravitational collapse of a generalised Vaidya spacetime in the context of the Cosmic Censorship hypothesis. We develop a general mathematical framework to study the conditions on the mass function so that future directed non-spacelike geodesics can terminate at the singularity in the past. Thus our result generalises earlier works on gravitational collapse of  the combinations of Type-I and Type-II matter fields. Our analysis shows transparently that there exist classes of generalised Vaidya mass functions for which the collapse terminates with a locally naked central singularity. We calculate the strength of the these singularities to show that they are strong curvature singularities and there can be no extension of spacetime through them.
\end{abstract}


\maketitle

\section{Introduction\label{Intro}}
The Vaidya spacetime \cite{Vaidya}, also known as the radiating Schwarzschild spacetime, describes the geometry outside a radiating spherically symmetric star. The radiation effects are important in the later stages of gravitational collapse of a star, when considerable amount of energy in form of photons or neutrinos is ejected from the star. This makes the collapsing star 
to be surrounded by an ever expanding zone of radiation. If we treat the complete non-static configuration of the radiating star and the zone of radiation as an isolated system within an asymptotically flat universe, then beyond the expanding zone of radiation the spacetime may be described by the Schwarszchild solution. The Vaidya solution is of Petrov type D and possesses a normal shear-free null congruence with non-zero expansion. In terms of exploding (imploding) null coordinates the metric is given as 
\begin{equation}
ds^2 = -\left[1-\frac{2m(v)}{r}\right ]dv^2+2\epsilon dvdr +r^2d\Omega^2,  \label{vaidya}
\end{equation}
where $\epsilon=\pm 1$ describes incoming (outgoing) radiation shells respectively, the function `$m(v)$' is the mass function and $d\Omega^2$ describes the line element on the 2-sphere. 

One of the earliest counterexamples of  Cosmic Censorship Conjecture (CCC), with a reasonable matter field satisfying physically reasonable energy conditions, was found in the shell focussing singularity formed by imploding shells of radiation in the Vaidya-Papapetrou model \cite{Papa,Dwivedi_1989}. In this case, radially injected radiation flows into an initially flat and empty region, and is focussed into a central singularity of growing mass. 
It was shown that the central singularity at ($v=0,r=0$) becomes a node with definite tangent for families of non-spacelike geodesics, for a non-zero measure of parameters in the model. Hence the singularity at ($v=0,r=0$) is naked in the sense that families of future directed non-spacelike geodesics going to future null infinity terminate at the singularity 
in the past. For a detailed discussion on the censorship violation in radiation collapse we refer to \cite{Joshi_1993}.

The generalisation of the Vaidya solution, also known as the generalised Vaidya spacetime, that includes all the known solutions of Einstein's field equations with combination of Type-I and Type-II matter fields, was given by Wang and Yu  \cite{Wang}. This generalisation comes from the observation that the energy momentum tensor for these matter fields are linear in terms of the mass function. As a result, the linear superposition of particular solutions is also a solution to the field equations. Hence, by superposition we can explicitly construct solutions such as 
the monopole-de Sitter-charged Vaidya solution and the Husain solution. Generalised Vaidya spacetimes are also widely used in describing the formation of regular black holes \cite{sean}, dynamical black holes \cite{Ghosh_2004} and black holes with closed trapped regions \cite{frolov}. The generalised Vaidya model can be matched to a heat conducting interior of a radiating star as recently shown by \cite{Maharaj}. Also, generalised Vaidya spacetime emerges naturally while solving many other astrophysical and cosmological scenarios \cite{Reza, Sungwook}.

The main goal of this paper is to study the gravitational collapse of generalised Vaidya spacetimes in the context of CCC. We develop a general mathematical framework to study the conditions on the mass function so that future directed non-spacelike geodesics can terminate at the singularity in the past. Thus our result generalises the earlier works on gravitational collapse of Type-II matter fields and also shows transparently that there exist classes of generalised Vaidya mass function (of non-zero measure) for which the collapse terminates with a locally naked singularity. We also calculate the strength of the naked singularities to show that they are strong curvature singularities and there is no extension of spacetime through these singularities.

The paper is organised as follows: In the next section we discuss the generalised Vaidya solution. In section 3 we consider a type II matter field undergoing gravitational collapse to a spacetime singlarity. In section 4 we analyse the central singularity (at $v=0,r=0$) to find out the conditions on the mass function such that the singularity is a node with a definite tangent for non-spacelike geodesics. In section 5 we discuss the strength of the naked central singularity. Finally in the last section we apply our result to some well known collapse models to recover the conditions for CCC violation.

Unless otherwise specified, we use natural units ($c=8\pi G=1$) throughout this paper, Latin indices run from 0 to 3. 
The symbol $\nabla$ represents the usual covariant derivative and $\partial$ corresponds to partial differentiation. 
We use the $(-,+,+,+)$ signature and the Riemann tensor is defined by
\begin{equation}
R^{a}{}_{bcd}=\Gamma^a{}_{bd,c}-\Gamma^a{}_{bc,d}+ \Gamma^e{}_{bd}\Gamma^a{}_{ce}-\Gamma^e{}_{bc}\Gamma^a{}_{de}\;,
\end{equation}
where the $\Gamma^a{}_{bd}$ are the Christoffel symbols (i.e. symmetric in the lower indices) defined by
\begin{equation}
\Gamma^{a}{}_{bd}=\frac{1}{2}g^{ae}
\left(g_{be,d}+g_{ed,b}-g_{bd,e}\right)\;.
\end{equation}
The Ricci tensor is obtained by contracting the {\em first} and the {\em third} indices
\begin{equation}\label{Ricci}
R_{ab}=g^{cd}R_{cadb}\;.
\end{equation}
 The Hilbert--Einstein action in the presence of matter is given by
\begin{equation}
{\cal S}=\frac12\int d^4x \sqrt{-g}\left[R-2\Lambda-2{\cal L}_m \right]\;,
\end{equation}
variation of which gives the Einstein field equations 
\be
G_{ab}+\Lambda g_{ab}=T_{ab}\;.
\ee

\section{Generalised Vaidya Spacetime \label{one}}

We know that the most general spherically symmetric line element for an arbitrary combination of Type-I matter fields (whose energy momentum tensor has one timelike and three spacelike eigenvectors) and Type-II matter fields (whose energy momentum tensor has double null eigenvectors) is given by \cite{israel}:
\begin{eqnarray}
ds^2& = &-e^{2\psi (v,r)}\left[1-\frac{2m(v,r)}{r}\right ]dv^2+2\epsilon e^{\psi (v,r)}dvdr \nonumber\\
&&+r^2(d\theta^2+\sin^2\theta d\phi ^2), \:(\epsilon=\pm 1).\label{generalized-vaidya}
\end{eqnarray}
Here $m(v,r)$ is the mass function related to the gravitational energy within a given radius $r$ \cite{Lake}. When $\epsilon=+1$ , the null coordinate $v$ represents the Eddington advanced time, where $r$ is decreasing towards the future along a ray $v=Const.$ and depicts ingoing null congruence while $\epsilon=-1$ depicts an outgoing null congruence.

The specific combination of matter fields that makes $\psi (v,r)= 0$ gives the generalised Vaidya geometry. In this paper, as we are considering a collapse scenario, we take $\epsilon =+1$. Particularly, we consider a line element of the form
\begin{equation}
ds^2 = -\left(1-\frac{2m(v,r)}{r}\right )dv^2+2dvdr +r^2d\Omega^2. \label{line-element}
\end{equation}
Using the following definitions
\begin{equation}
\dot{m}(v,r)\equiv \frac{\partial m(v,r)}{\partial v},  \: \quad m'(v,r)\equiv  \frac{\partial m(v,r)}{\partial r},
\end{equation}
the non-vanishing components of the Ricci tensor can be written as
\begin{subequations}\label{Ricci}
\begin{eqnarray}
  R^v_v &=& R^r_r=\frac{m''(v,r)}{r}\label{Ricci1} \,,\\
 R^{\theta}_{\theta} &=& R^{\phi}_{\phi}=\frac{2m'(v,r)}{r^2}\,,\label{Ricci2}
\end{eqnarray}
\end{subequations}\label{Einstein}
while the Ricci scalar is given by
\begin{equation}\label{scalarcurvature}
    R = \frac{2m''(v,r)}{r} + \frac{4m'(v,r)}{r^2}.
\end{equation}
The non-vanishing components of the Einstein tensor can be written as
\begin{subequations}
\begin{eqnarray}
G^v_v &=& G^r_r=-\frac{2m'(v,r)}{r^2},\label{Einstein1}\\
G^r_v &=& \frac{2\dot{m} (v,r)}{r^2},\label{Einstein2}\\
G^{\theta}_{\theta} &=& G^{\phi}_{\phi}=-\frac{m''(v,r)}{r}\;.\label{Einstein3}
\end{eqnarray}
\end{subequations}
Using the Einstein field equations, the corresponding energy momentum tensor can be written in the form \cite{Husian_1996, Wang}
\begin{equation}
T_{\mu\nu}=T^{(n)}_{\mu\nu}+T^{(m)}_{\mu\nu}, \label{EMT}
\end{equation}
where
\begin{subequations}
\begin{eqnarray}\label{EMT2}
T^{(n)}_{\mu\nu}&=&\mu l_{\mu}l_{\nu},\\
T^{(m)}_{\mu\nu}&=&(\rho+\varrho)(l_{\mu}n_{\nu}+l_{\nu}n_{\mu} )+ \varrho g_{\mu\nu},
\end{eqnarray}
\end{subequations}
and
\begin{eqnarray}
\mu = \frac{2\dot{m}(v,r)}{r^2}, \quad \rho=\frac{2m'(v,r)}{r^2}, \quad \varrho = -\frac{m''(v,r)}{r}.
\end{eqnarray}
In the above $l_{\mu}$ and $n_{\mu}$ are two null vectors
\begin{eqnarray}
l_{\mu} =\delta ^0_{\mu}, \quad  n_{\mu}=\frac{1}{2}\left [1-\frac{2m(v,r)}{r}\right]\delta ^0_{\mu}-\delta^1_{\mu},
\end{eqnarray}
where $l_{\mu}l^{\mu} =n_{\mu}n^{\mu}=0$ and $l_{\mu}n^{\mu}=-1$.

Equation \eqref{EMT} can be considered as the energy momentum tensor of the generalised Vaidya solution, with the component $T^{(n)}_{\mu\nu}$ is the matter field that moves along the null hypersurfaces $v = Const.,$ while $T^{(m)}_{\mu\nu}$ describes the matter moving along timelike trajectories. When $\rho = \varrho = 0$, the solution reduces to the Vaidya solution with $m = m(v)$.

If the energy momentum tensor of Equation \eqref{EMT} is projected to the orthonormal basis, defined by the four vectors,
\begin{eqnarray}\label{orthonormal}
    E^{\mu}_{(0)} = \frac{l_{\mu} + n_{\mu}}{\sqrt{2}}, &&\quad E^{\mu}_{(1)} = \frac{l_{\mu} - n_{\mu}}{\sqrt{2}},\nonumber \\
     \quad E^{\mu}_{(2)} = \frac{1}{r}\delta^{\mu}_2, &&\quad E^{\mu}_{(3)} = \frac{1}{r\sin\theta}\delta^{\mu}_3,
\end{eqnarray}
it can be shown that \cite{Wang}
\begin{equation}\label{EMTMatrix}
    \left[T_{(\mu)(\nu)}\right] = \left[
                                \begin{array}{cccc}
                                  \frac{\mu}{2}+ \rho & \frac{\mu}{2} & 0 & 0 \\
                                  \frac{\mu}{2} & \frac{\mu}{2} - \rho & 0 & 0 \\
                                  0 & 0 & \varrho & 0 \\
                                  0 & 0 & 0 & \varrho \\
                                \end{array}
                              \right].
\end{equation}
This form of the energy momentum is a combination of Type-I and Type-II fluids as defined in \cite{Ellis_1973}, with the following energy conditions
\begin{enumerate}[a)]
  \item \emph{The weak and strong energy conditions}
  \begin{equation}\label{strongweakenergy}
   \mu\geq 0, \quad \rho\geq 0, \quad \varrho \geq 0, \quad (\mu\neq 0).
  \end{equation}
  \item \emph{The dominant energy conditions}
  \begin{equation}\label{dominantenergy}
    \mu\geq 0, \quad \rho \geq \varrho \geq 0, \quad (\mu\neq 0).
  \end{equation}
\end{enumerate}
These energy conditions can be satisfied by choosing the mass function $m(v,r)$ suitably. In particular, when
$m = m(v)$, all the energy conditions (weak, strong, and dominant) reduce to $\mu \geq 0$, while when $m = m(r)$ we have $\mu = 0$, and the matter field degenerates to a Type-I fluid with the usual energy conditions \cite{Ellis_1973} .

\section{Collapsing Model \label{two}}
In this section, we examine the gravitational collapse of  imploding radiation and matter described by the generalised Vaidya spacetime. In this situation, a thick shell of radiation and Type-I matter collapses at the centre of symmetry \cite{Joshi_1993}.

If $K^{\mu}$ is the tangent to non-spacelike geodesics with $K^{\mu} = \frac{dx^{\mu}}{dk}$, where $k$ is the affine parameter, then $K^{\mu}_{;\nu}K^{\nu}=0$ and
\begin{equation}
g_{\mu\nu}K^{\mu}K^{\nu} =\beta, \label{tangents}
\end{equation}
where $\beta$ is a constant that characterises different classes of geodesics. $\beta=0$ describes  null geodesics, while $\beta < 0$ applies to timelike geodesics.  
The equations  for $dK^v/dk$  and $dK^r/dk$ are calculated from the Euler-Lagrange equations 
\begin{equation}
\frac{\partial L}{\partial x^\alpha}- \frac{d}{dk}\left(\frac{\partial L}{\partial \dot{x^\alpha} }\right)= 0, 
 \end{equation}
with the Lagrangian 
\begin{equation}
L = \frac{1}{2}g_{\mu\nu}\dot{x}^{\mu}\dot{x}^{\nu}, 
 \end{equation}
where the dot is a derivative with respect to the affine parameter $k$. These equations are given by
\begin{widetext}
\begin{subequations}\label{tangents}
\begin{eqnarray}
\label{tangenteqn1}\frac{dK^v}{dk}+\left(\frac{m(v,r)}{r^2}-\frac{m'(v,r)}{r}\right)\left(K^v\right)^2 -\frac{\ell^2}{r^3} &=& 0,\quad \quad \quad\\
\frac{dK^r}{dk}+\frac{\dot{m}(v,r) }{r}\left(K^v\right)^2 - \frac{\ell^2}{r^3}\left(1 - \frac{2m(v,r)}{r}\right) - \beta\left(\frac{m(v,r)}{r^2}- \frac{m'(v,r)}{r}\right) &=& 0. \label{tangenteqn2}
\end{eqnarray}
\end{subequations}
\end{widetext}
The components $K^{\theta}$ and $K^{\phi}$ of the tangent vector are given by\cite{Joshi_1993}
\begin{subequations}\label{thetaphitangents}
\begin{eqnarray}
K^{\theta}&=&\frac{\ell \cos\varphi}{r^2\sin^2\theta},\label{thetatangent}\\
K^{\phi}&=& \frac{\ell \sin\varphi \cos\phi }{r^2}\label{phitangent},
\end{eqnarray}
\end{subequations}
where $\ell $ is the impact parameter and $\varphi$ is the isotropy parameter defined by  the relation $\sin\phi \tan\varphi = \cot\theta$.

If we follow \cite{Dwivedi_1989}  and write $K^v$ as
\begin{equation}\label{kvcomponent}
K^v=\frac{P}{r},
\end{equation}
where $P=P(v,r)$ is an arbitrary function, 
then $g_{\mu\nu}K^{\mu}K^{\nu} =\beta$ gives
\begin{eqnarray}
K^r &=& \frac{P}{2r}\left[1 - \frac{2m(v,r)}{r}\right] - \frac{\ell^2}{2rP}+ \frac{\beta r}{2P}\;.\label{rcomponent}
\end{eqnarray}
From Equation \eqref{kvcomponent}, we have
\begin{equation}\label{dKV}
\frac{dK^v}{dk} = \frac{d}{dk}\left(\frac{P}{r}\right) = \frac{1}{r}\frac{dP}{dk} - \frac{P}{r^2}\frac{dr}{dk}.
\end{equation}
Thus
\begin{equation}\label{none}
\frac{dP}{dk} = \frac{1}{r}\left(r^2\frac{dK^v}{dk} + P\frac{dr}{dk}\right).
\end{equation}
Substituting Equations \eqref{tangenteqn1} and \eqref{rcomponent} into Equation \eqref{none} gives the differential equation satisfied by the function $P$:
\begin{equation}
\frac{dP}{dk} = \frac{P^2}{2r^2}\left(1 - \frac{4m(v,r)}{r} + 2m'(v,r)\right) + \frac{\ell^2}{2r^2} + \frac{\beta}{2}\;.
\label{Pdifferential}
\end{equation}
The function $P(v,r)$ can be found if the mass function $m(v,r)$ and the initial conditions are defined (See for example, \cite{Dwivedi_1989}).

\section{Conditions for Locally Naked Singularity}\label{three}

In this section we examine, given the generalised Vaidya mass function, how the final fate of collapse
is determined in terms of either a black hole or a naked singularity. If there are families of future directed non-spacelike
trajectories reaching faraway observers in spacetime, which terminate in the past at the singularity, then we have a
naked singularity forming as the collapse final state. Otherwise when no such families exist and event
horizon forms sufficiently early to cover the singularity, we have a black hole.
The equation for the radial null geodesics $(\ell = 0, \beta = 0)$ for the line element \eqref{line-element} can be easily found, using equations \eqref{kvcomponent} and \eqref{rcomponent}, which is given by
\begin{equation}\label{nullgeodesics}
    \frac{dv}{dr} = \frac{2r}{r - 2m(v,r)}.
\end{equation}
The above differential equation has a singularity at $r = 0, \quad v = 0$. The nature of this singularity can be analysed by the usual techniques of the theory of ODE's \cite{Tricomi_1961,Perko_1991}. Whereas the procedures used below are standard, we shall describe the case treated here in some detail
so as to give the exact picture of the nature of the central singularity at $r = 0, \quad v = 0$.

\subsection{Structure of the central singularity}

We can generally write Equation \eqref{nullgeodesics} in the form
 \begin{equation}\label{Abridgeneral}
 \frac{dv}{dr} = \frac{M(v,r)}{N(v,r)},
  \end{equation}
 with the singular point at $r = v = 0$, where both the functions $M(v,r)$ and $N(v,r)$ vanishe. Hence we should carefully analyse the existence and uniqueness of the solution of the above differential equation in the vicinity of this singularity. At this point it is useful to introduce a new independent variable $t$ with differential $dt$ such that
\begin{equation}\label{dtdifferential1}
    \frac{dv}{M(v,r)} = \frac{dr}{N(v,r)} = dt,
\end{equation}
so that the  differential equation \eqref{Abridgeneral} can be replaced by a system
\begin{eqnarray}\label{system}
 \nonumber   \frac{dv(t)}{dt} &=& M(v,r) \\
              \frac{dr(t)}{dt} &=& N(v,r)\;.
\end{eqnarray}
We would like to emphasise here that all the solutions of equation (\ref{Abridgeneral}) is a solution of the system (\ref{system}) and hence we study the behaviour of this system of equations near the singular point $r = v = 0$ in the ($r,v$) plane.
We can easily see that the singular point of (\ref{Abridgeneral}) is a fixed point of the system (\ref{system}). To find the necessary and sufficient conditions for existence of the solutions of this system in the vicinity of the fixed point $r = v = 0$, let us write (\ref{system}) as a differential equation of the vector $\bm{x}(t)=[v(t),r(t)]^T$ on  $\mathbb{R}^2$ as
\begin{equation}
\frac{d{\bm{x}(t)}}{dt}=\bm{f}(\bm{x}(t))
\label{vec_sys}
\end{equation}
Now to show the existence and uniqueness of the solution with respect to the initial conditions arbitrarily near the fixed point of the above system (since the initial conditions on the fixed point will imply the system stays on the fixed point) we give the following definitions:
\begin{defn}
The function $\bm{f}:\mathbb{R}^2\rightarrow\mathbb{R}^2$ is differentiable at $\bm{x}=\bm{x}_0$, if the partial derivatives of the functions $M$ and $N$ with respect to $r$ and $v$ exist at that point. The derivative of the function, $\bm{Df}$, is given by the $2\times2$ Jacobian matrix\\
\begin{displaymath}
 \left[
                                \begin{array}{cc}
                                  M_{, v}& M_{,r} \\
                                  N_{,v} & N_{,r} \\
                                \end{array}
                              \right].
\end{displaymath} 
\end{defn}
\begin{defn}
Suppose $U$ is an open subset of  $\mathbb{R}^2$, then $\bm{f}:U\rightarrow\mathbb{R}^2$ is of class $C^1$ iff the partial derivatives $M_{, v}, M_{,r} , N_{,v} ,N_{,r} $ exist and are continuous on $U$.
\end{defn}
Henceforth we will consider the function $\bm{f}$ to be of class $C^1$ throughout the spacetime. Let us now show that there exists a unique solution to the system (\ref{vec_sys}) subject to the initial condition $\bm{x}(t_0)=\bm{x}_0$, where $\bm{x}_0$ is arbitrarily near the fixed point of the equation. Let us define an operator $T$ in the following way:
\begin{defn}
Let $T:\mathbb{R}^2\rightarrow\mathbb{R}^2$ be an operator acting on all continuous and differentiable vectors $\bm{y}(t)$ on $\mathbb{R}^2$ and takes them to the image $T\bm{y}(t)$ defined as
\begin{displaymath}
T\bm{y}(t)=\bm{x}_0+\int_{t_0}^t \bm{f}(\bm{y}(s))ds\;.
\end{displaymath} 
\end{defn}
We now prove an important property of this operator $T$, subject to the function $\bm{f}$ being class $C^1$,
\begin{lem}
Let $U\ni \bm{x}_0$ be an open subset of  $\mathbb{R}^2$ and $\bm{f}:U\rightarrow\mathbb{R}^2$ is of class $C^1$ and $\bm{y}(t)$, $\bm{z}(t)$ are continuous and differentiable vectors on $U$. Then there always exists an $\epsilon$-neighbourhood $B_\epsilon(\bm{x}_0)$ of $\bm{x}_0$ in which $|T\bm{y}(t)-T\bm{z}(t)|\le \kappa |\bm{y}(t)-\bm{z}(t)|$ where 
$0\le\kappa\le 1$. In other words $T$ is an contraction mapping on $B_\epsilon(\bm{x}_0) $.
\end{lem}
\begin{proof}
Let $K_0={\rm max}_{|\bm{x}-\bm{x}_0|\le\epsilon}||\bm{Df}(\bm{x})||$. Then we have 
\be
|T\bm{y}(t)-T\bm{z}(t)|=\left|\int_{t_0}^t \left(\bm{f}(\bm{y}(s))-\bm{f}(\bm{z}(s))\right)ds\right |\;.
\ee
The above equation can be written as 
\be
|T\bm{y}(t)-T\bm{z}(t)|= \left|\int_{t_0}^t \left(\int_{\bm{z}(s)}^{\bm{y}(s)}\bm{Df}(\bm{r})dr\right)ds\right |\;,
\ee
and therefore we get the inequality
\be
|T\bm{y}(t)-T\bm{z}(t)|\le K_0|(t-t_0)|\,|\bm{y}(t)-\bm{z}(t)|\;.
\ee
Hence there always exists an open interval $(t_0-h, t_0+h)$ (that corresponds to a neighbourhood around $\bm{x}_0$) where $K_0|(t-t_0)|\le1$ and $T$ is a contraction mapping.
\end{proof}

Having established the existence of a contraction mapping in a neighbourhood of the point $\bm{x}_0$ and recalling that $\mathbb{R}^2$ is a complete metric space, we now use the following theorem to establish the existence and uniqueness of the solution of the system (\ref{vec_sys}) subject to the initial condition $\bm{x}(t_0)=\bm{x}_0$.
\begin{thm}
If $T:\mathbb{X}\rightarrow\mathbb{X}$ is a contraction mapping on a complete metric space $\mathbb{X}$, then there is exactly one solution of the equation $T\bm{x}=\bm{x}$.
\end{thm}
The above theorem establishes a unique solution of the system (\ref{vec_sys}) with the initial condition $\bm{x}(t_0)=\bm{x}_0$ in an $\epsilon$-neighbourhood of the point $\bm{x}_0$ which is given by
\be
\bm{x}(t)=\bm{x}_0+\int_{t_0}^t \bm{f}(\bm{x}(s))ds\;.
\label{sol1}
\ee
The assumption that $\bm{f}:U\rightarrow\mathbb{R}^2$ is of class $C^1$ assures the solution to be continuous and differentiable in this neighbourhood. 
Let us now find the nature of the fixed point $r = v = 0$ of the system (\ref{vec_sys}). As the partial derivatives of the functions $M$ and $N$ exist and are continuous in the neighbourhood of the fixed point, we can linearise the system near the fixed point and hence the general behaviour of this system  near the singular point is similar to the characteristic equations  \cite{Tricomi_1961}
 \begin{eqnarray}\label{linearsystem}
 \nonumber   \frac{dv}{dt} &=& Av + Br \\
              \frac{dr}{dt} &=& Cv + Dr,
\end{eqnarray}
where $A = \dot{M}(0,0)$, $B = M'(0,0)$, $C = \dot{N}(0,0)$, $N'(0,0) = D$, with the dot denoting partial differentiation with respect to the variable $v$ while the dash denotes partial differentiation with respect to the coordinate $r$ and $AD - BC \neq 0$. By using a linear substitution of the type
\begin{eqnarray}
\nonumber  \xi &=& \alpha v + \omega r\\
 \label{linearsub} \eta &=& \gamma v+ \delta r,
\end{eqnarray}
where $\alpha\delta - \omega\gamma \neq 0,$ and the equation
\begin{equation}\label{eqnsub}
    \frac{d\eta}{d\xi} = \frac{\chi_2\eta}{\chi_1\xi},
\end{equation}
the system \eqref{linearsystem} can be reduced into the form
\begin{eqnarray}
 \nonumber \frac{d\xi}{dt} &=& \chi_1 \xi\\
 \label{reducedeqn} \frac{d\eta}{dt} &=& \chi_2\eta.
\end{eqnarray}
Using equations \eqref{linearsystem}, \eqref{linearsub} and \eqref{reducedeqn}, it can be found that
\begin{eqnarray*}
  \alpha(Av + Br) + \omega(Cv + Dr) &=& \chi_1(\alpha v + \omega r)\\
  \gamma(Av + Br) + \delta(Cv + Dr) &=& \chi_2(\gamma v + \delta r).
\end{eqnarray*}
By equating the coefficients of $v$ and $r$ in the above equations, we obtain
\begin{eqnarray}
  \nonumber (A - \chi_1)\alpha + C\omega &=& 0 \\
  B\alpha + (D - \chi_1)\omega &=& 0
\end{eqnarray}
and
\begin{eqnarray}
 \nonumber (A - \chi_2)\gamma + C\delta &=& 0 \\
  B\gamma + (D - \chi_2)\delta &=& 0.
\end{eqnarray}
The above equations in $\alpha, \; \omega$ and $\gamma, \; \delta$ may be satisfied by the values of  $\alpha, \; \omega, \; \gamma, \; \delta$ not all zero if the determinant of the coefficients is zero. That is
\begin{equation}\label{determinant}
  \left| \begin{array}{cc}
            A - \chi & C \\
            B & D - \chi \\
          \end{array}
  \right| = 0,
\end{equation}
or
\begin{equation}\label{characteristiceqn}
    \chi^2 - (A + D)\chi + AD - BC = 0.
\end{equation}
This is the characteristic equation with roots (eigenvalues) $\chi_1$ and $\chi_2$ given by
\begin{equation}\label{characteristicsoln}
    \chi = \frac{1}{2}\left((A + D) \pm \sqrt{(A - D)^2 + 4BC}\right).
\end{equation}
The singularity of Equation \eqref{linearsystem} is classified as a node if $(A - D)^2 + 4BC \geq 0$ and $BC > 0$. Otherwise, it may be a  centre or focus.

Now, for the equation  \eqref{nullgeodesics} we have  $M(v,r) = 2r$, $ N(v,r) = r -2m(v,r)$.
If at the central singularity, $v = 0$, $r = 0$, we define the following limits
\begin{subequations}
\begin{eqnarray}
m_0& =& \lim\limits_{v\to 0, r \to 0}m(v,r),\\ \dot{m}_0 &=& \lim\limits_{v\to 0, r \to 0}\frac{\partial}{\partial v}m(v,r), \\m'_0 &= &\lim\limits_{v\to 0, r \to 0}\frac{\partial}{\partial r}m(v,r),
\end{eqnarray}
\end{subequations}
then the null geodesic equation can be linearised near the central singularity  as
\begin{equation}\label{abridgednew2}
    \frac{dv}{dr} = \frac{2r}{(1 - 2m'_0)r - 2\dot{m}_0v}.
\end{equation}
Clearly, this equation has a singularity ar $v = 0, \quad r = 0$. We can determine the nature of this singularity by observing the value of the discriminant of the characteristic equation. Using equation \eqref{characteristicsoln}, the roots of the characteristic equation are given by
\begin{equation}\label{newcharacteristicsoln}
    \chi = \frac{1}{2}\left((1 - 2m'_0)\pm \sqrt{(1 - 2m'_0)^2 - 16\dot{m}_0}\right).
\end{equation}
For the singular point at $r = 0, v = 0$ to be a node, it is required that 
\begin{equation}\label{Conditions}
(1 - 2m'_0)^2 - 16\dot{m}_0 \geq 0 \quad \text{ and} \quad  \dot{m}_0 > 0.
 \end{equation}
Thus, if the mass function $m(v,r)$ is chosen such that the condition in Equation \eqref{Conditions} is satisfied, then
the singularity at the origin $(v = 0, r = 0)$ will be a node and outgoing non-spacelike geodesics can come out of the singularity with a definite value of the tangent. 

\subsection{Existence of outgoing non-spacelike geodesics}

Let us now return to the physical problem of the collapsing generalised Vaidya spacetime and let us choose the mass function that has the following properties:
\begin{enumerate}
\item The mass function $m(v,r)$ obeys all the physically reasonable energy conditions throughout the spacetime.
\item The partial derivatives of the mass function exist and are continuous on the entire spacetime.
\item The limits of the partial derivatives of the mass function $m(v,r)$ at the central singularity obey the conditions: $(1 - 2m'_0)^2 - 16\dot{m}_0 \geq 0$ and $\dot{m}_0 > 0$.
\end{enumerate}
The choice of the mass function with the above properties would ensure the existence and uniqueness of the solutions of the null geodesic equation in the vicinity of the central singlarity and will also make the central singularity a node of  $C^1$ solutions with definite tangents.

To find the condition for the existence of outgoing radial non-spacelike geodesics from the nodal singularity, we consider the tangent of these curves at the singularity. Suppose $X$ denotes the tangent to the radial null geodesic. If the limiting value of $X$ at the singular point is positive and finite then we can see that outgoing future directed null geodesics do terminate in the past at the central singularity. 
The existence of these radial null geodesics  characterises the nature (a naked singularity or a black hole) of the collapsing solutions. In order to determine the nature of the limiting value of $X$ at $r = 0, v = 0$ we define
\begin{equation}\label{limitingvalue}
   X_0 = \lim_{\substack{v\rightarrow 0, r\rightarrow 0}}X = \lim_{\substack{v\rightarrow 0, r\rightarrow 0}}\frac{v}{r}.
\end{equation}
Using Equation \eqref{abridgednew2} and L'Hospital's rule (for the $C^1$ null geodesics) we get
\begin{equation}\label{limitingvalueX0}
   X_0 = \lim_{\substack{v\rightarrow 0, r\rightarrow 0}}\frac{v}{r} = \frac{dv}{dr} = \frac{2}{\left(1-2m'_0\right) - 2\dot{m}_0(\frac{v}{r})} ,
\end{equation}
which simplifies to
\begin{equation}\label{limitingvalueX02}
   X_0 = \frac{2}{(1-2m'_0) - 2\dot{m}_0X_0} .
\end{equation}
Solving for $X_0$ gives
\begin{equation}\label{X_0 value}
    X_0 = b_{\pm} = \frac{(1-2m'_0)\pm\sqrt{(1-2m'_0)^2 - 16\dot{m}_0}}{4\dot{m}_0}.
\end{equation}
If we can get one or more positive real roots by solving equation \eqref{limitingvalueX02}, then the singularity may be locally naked if the null geodesic lies outside the trapped region. In the next subsection we will calculate the dynamics of the trapped region to find the conditions for the existence of such geodesics.  

 \subsection{Apparent horizon}

The occurrence of naked singularity or black hole is usually decided by casual behaviour of the trapped surfaces developing in the spacetime during the collapse evolution.
The apparent horizon is the boundary of trapped surface region in the spacetime. For the generalised Vaidya spacetime the equation of the apparent horizon is given as
\begin{equation} \label{apparent:eqn}
\frac{2m(v,r)}{r} = 1.
\end{equation}

Thus, the slope of the apparent horizon can be calculated in the following way: we know 
\begin{subequations}
\begin{equation} \label{slope1:apparent}
\frac{2dm(v,r)}{dr} = 1,
\end{equation}
\begin{equation}\label{slope2:apparent}
2\left(\frac{\partial m}{\partial v}\right)\left(\frac{dv}{dr}\right)_{AH} + \frac{2\partial m}{\partial r} = 1,
\end{equation}
\end{subequations}
which finally gives the slope of the apparent horizon at the central singularity ($v \to 0, r \to 0$) as 
\begin{equation}\label{slopefianl:apparent}
\left(\frac{dv}{dr}\right)_{AH} = \frac{1 - 2m'_0}{2\dot{m}_0}.
\end{equation}
Thus now we have the sufficient conditions for the existence of a locally naked central singularity for a collapsing generalised Vaidya spacetime, which we state in the following proposition:
\begin{prop}
Consider a collapsing generalised Vaidya spacetime from a regular epoch, with a mass function $m(v,r)$ that obeys all the physically reasonable energy conditions and is differentiable in the entire spacetime. If the following conditions are satisfied :
\begin{enumerate}
\item The limits of the partial derivatives of the mass function $m(v,r)$ at the central singularity obey the conditions: $(1 - 2m'_0)^2 - 16\dot{m}_0 \geq 0$ and $\dot{m}_0 > 0$,
\item There exist one or more positive real roots $X_0$ of the equation (\ref{X_0 value}),
\item At least one of the positive real roots is less than $\left(\frac{dv}{dr}\right)_{AH}$ at the central singularity,
\end{enumerate}
then the central singularity is locally naked with outgoing $C^1$ radial null geodesics escaping to the future.
\end{prop}

\begin{figure}[!h]
\centering
\includegraphics[scale=0.45]{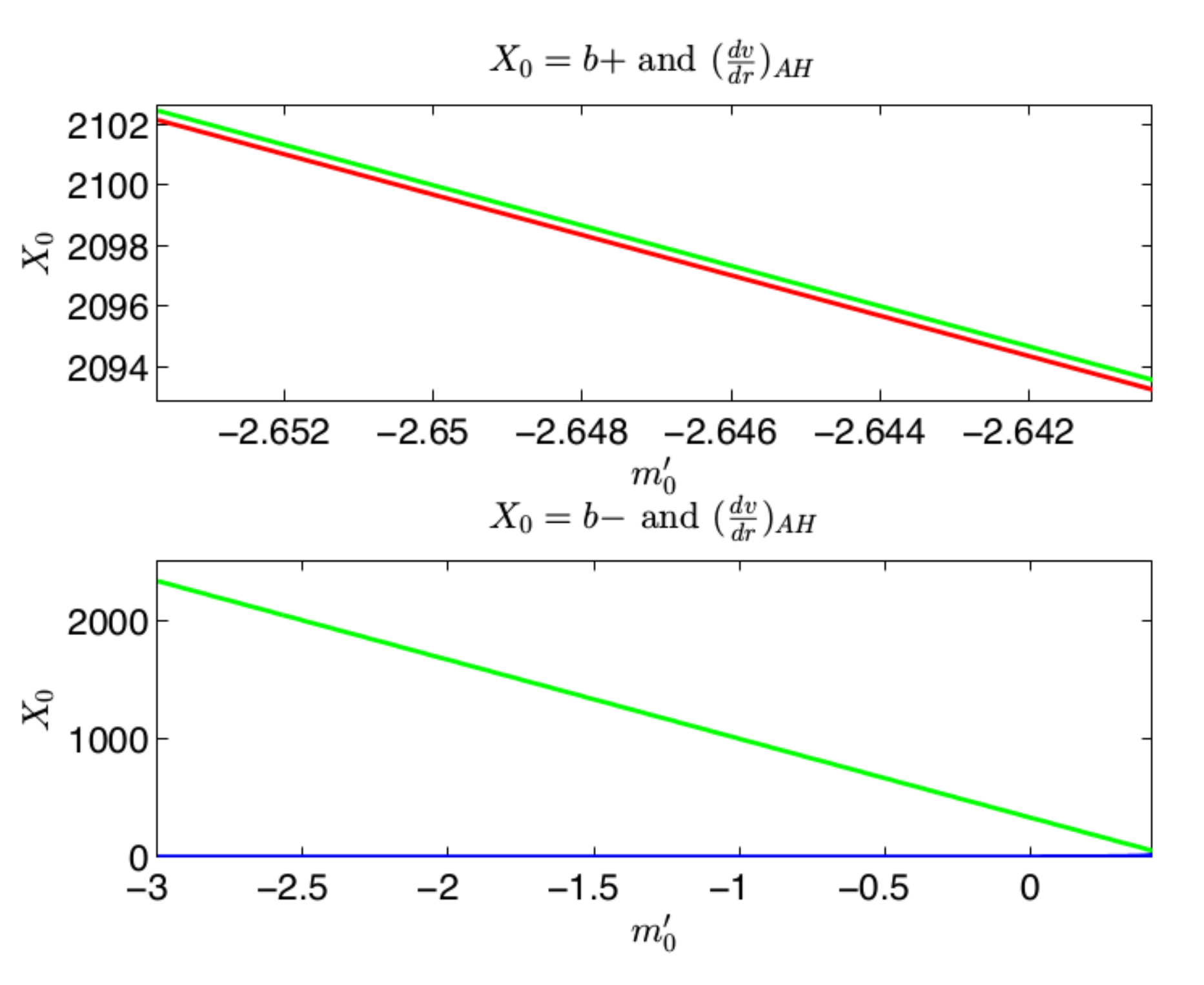}
  \caption{Variation of $X_0$ (red and blue) and $\left(\frac{dv}{dr}\right)_{AH}$ (green) with $m'_0$ at fixed value $\dot{m}= 0.0015$}\label{variation1}
\end{figure}
Figure \ref{variation1} shows the values of $X_0$ and $\left(\frac{dv}{dr}\right)_{AH}$ when $m'_0$ is varied in the interval $-3.0 \leq m'_0 \leq 0.42$ for a fixed value of $\dot{m}_0$. It can be observed from the figure that the value of $X_0 = b_{\pm}$  is always below the value of $\left(\frac{dv}{dr}\right)_{AH}$, and thus there exist open sets of parameter values for which the singularity is locally naked.
\section{Strength of Singularity}
To compute the strength of singularity according to Tipler \cite{Tipler}, which is the measure of its destructive capacity in the sense that whether extension of spacetime is possible through them or not \cite{Ghosh_2001}, we consider the null geodesics parameterized by the affine parameter $k$ and terminating at the shell focusing singularity $r = v = k = 0$. Following Clarke and Krolack \cite{Clarke}, a singularity would be strong if the condition
\begin{equation}\label{curvature}
  \lim_{\substack{k \rightarrow 0}}k^2\psi = \lim_{\substack{k \rightarrow 0}}k^2R_{\mu\nu}K^{\mu}K^{\nu} > 0,
\end{equation}
as defined by Tipler \cite{Tipler} (which is the sufficient condition for the singularity to be Tipler strong) where $R_{\mu\nu}$ is the Ricci tensor, is satisfied. We find the scalar $\psi = R_{\mu\nu}K^{\mu}K^{\nu}$ using equations \eqref{Ricci} \eqref{kvcomponent} and \eqref{rcomponent} as
\begin{equation}\label{Scalarcurvature}
   \psi = (2\dot{m}_0)\left(\frac{P}{r^2}\right)^2,
\end{equation}
and therefore,
\begin{equation}\label{ksquare}
    k^2\psi = (2\dot{m}_0)\left(\frac{Pk}{r^2}\right)^2.
\end{equation}
Using equations \eqref{kvcomponent}, \eqref{rcomponent} and L'Hospital's rule, we can evaluate the limit along non-spacelike geodesics as $k\rightarrow 0$. This limit is found to be
\begin{equation}\label{nonamegiven}
    \lim_{\substack{k \rightarrow 0}}k^2\psi = (2\dot{m}_0)\lim_{\substack{k \rightarrow 0}}\left(\frac{Pk}{r^2}\right)^2\;.
\end{equation}
 If we assume that $P \neq 0, \; \infty$, then by using L'Hospital's rule we have
\begin{equation}
    \lim_{\substack{k \rightarrow 0}}\left(\frac{Pk}{r^2}\right) = \lim_{\substack{k \rightarrow 0}}\left(\frac{Pdk}{2rdr}\right). 
\end{equation}
From equation \eqref{kvcomponent}, $\frac{P}{r} = \frac{dv}{dk}$. Therefore
\begin{equation}\label{limitscont.}
    \lim_{\substack{k \rightarrow 0}}\left(\frac{Pk}{r^2}\right) = \frac{1}{2}\frac{dv}{dk}\frac{dk}{dr} = \frac{1}{2}\frac{dv}{dr} =\frac{1}{2}X_0.
\end{equation}
Thus, we finally get
\begin{equation}\label{limitfinal}
    \lim\limits_{k\to 0}k^2\psi = \frac{1}{4}X_0^2(2\dot{m}_0).
\end{equation}
We observe that  the strength of the central singularity depends only on the limit of the derivative of mass function with respect to $v$ and the limiting value $X_0$.

With the suitable choice of the mass function (see Table \ref{Table_X_0} for some special cases), it can be shown that
\begin{equation}\label{canbeshown}
     \lim\limits_{k\to 0}k^2\psi  = \frac{1}{4}X_0^2(2\dot{m}_0) > 0.
\end{equation}
If this condition is satisfied for some real and positive root $X_0$, then we conclude that the observed naked singularity is strong. It is interesting to note that when the energy conditions are satisfied, then if a naked singularity is developed as a end state of the collapse, then that naked singularity is always strong.

\begin{table*}
\caption{Equations of tangents $X_0$ to the singularity curve and values of $\lim\limits_{k\to 0}k^2\psi$ for some special sub-classes of generalised Vaidya spacetime}\label{Table_X_0}
\begin{ruledtabular}
\begin{tabular}{lll}
    Spacetime & Equation for tangent to the singularity curve $X_0$ & $\lim\limits_{k\to 0}k^2\psi$ \\
    \hline
    Vaidya & $X_0 = \frac{2}{1-\lambda X_0}$ or $ X_0 = \frac{1\pm\sqrt{1 - 8\lambda}}{2\lambda},\quad 0<\lambda \leq \frac{1}{8} $ & $ \frac{1}{4}\lambda X_{0}^{2}$\\
    Charged Vaidya &$\mu^2X_0^3 - 2\lambda X_0^2 + X_0 -2 = 0$ & $\frac{1}{2} X_{0}^{2}\left(\lambda - \mu^2X_0\right)$\\
    Charged Vaidya-de Sitter & $\mu^2X_0^3 - 2\lambda X_0^2 + X_0 - 2 = 0$ & $\frac{1}{2} X_{0}^{2}\left(\lambda - \mu^2X_0\right)$\\ 
    Husain solution & $2\mu^2\left(1 - \frac{2k}{2k-1}\right)X_0^{2k+1} + \lambda X_0^2 - 2X_0 + 4 = 0$ & $\frac{1}{4}X_{0}^{2}\left(\lambda -\frac{4k\mu^2}{2k-1}X_0^{2k-1}\right)$\\
\end{tabular}
\end{ruledtabular}
\end{table*} 

\section{Some special sub-classes of Generalised Vaidya}
Using equation \eqref{limitingvalueX02}  we calculate the equations of tangents to the null geodesics at the central singularity for some special sub-classes of the generalised Vaidya spacetime with the specific mass function, $m(v,r)$. In all these mass functions, we can see that it is possible to obtain at least one or more real and positive value of $X_0$.
\begin{enumerate}[i.]
\item {\it The self-similar Vaidya spacetime}

In this case we consider the situation of a radial influx of null fluid in an initially empty region of Minkowski spacetime \cite{Joshi_1993,Naresh_2001}. 
The first shell arrives at $r = 0$ at time $v = 0$ and the final shell at $v=T$. 
A central singularity of the collapsing mass is developed at $r = 0$. For $v<0$ we have $m(v,r) = 0$ and for $v>T$ we have $m(v,r) = M_0$ where $M_0$ is the constant Schwarzschild mass. For the weak energy conditions to be satisfied, it is required that $\dot{m}(v,r)$ to be a non-negative. We define the mass function as
\begin{equation}\label{Vaidyasoln}
  m(v,r) = m(v),
\end{equation}
where
\begin{equation}
m(v) =
\begin{cases}
 0,& v < 0, \\
\frac{1}{2}\lambda v, & 0\leq v \leq T,\\
M_0, & v > T.
\end{cases}
\end{equation}
The mass function is a non-negative increasing function of $v$ for imploding radiation. For $ 0\leq v \leq T$, the solution is the self-similar Vaidya spacetime. For this choice of mass function, using equation \eqref{limitingvalueX02}, we get
\begin{equation}\label{X_0vaidya}
   X_0 = \frac{2}{1-\lambda X_0}\quad \text{or}\quad  X_0 = \frac{1\pm\sqrt{1 - 8\lambda}}{2\lambda}.
\end{equation}
This is similar to the solution obtained by Joshi \cite{Joshi_1993}. This equation gives positive values of $X_0$ for all values of $\lambda$ in the range $0<\lambda \leq \frac{1}{8}$. It can also be observed that $\lim\limits_{k\to 0}k^2\psi = \frac{1}{4}\lambda X_{0}^{2} > 0$ for all positive values of $X_0$; hence the singularity is strong.

\item {\it The charged Vaidya spacetime}  

This subclass of the generalised Vaidya spacetime has been studied in great detail in \cite{Lindquist_1965, Israel_1967, Patil_199}.   We consider here the form of the  mass function in  \cite{Wang, Naresh_2001}
\begin{subequations}\label{chargedvaidya}
\begin{equation}
    m(v,r) = f(v) - \frac{e^2(v)}{2r},
\end{equation}
where $f(v)$ and $e(v)$ are arbitrary functions representing the mass and electric charge respectively (limited only by the energy conditions), at the advanced time $v$. Particularly, we define these functions as \cite{Beesham}
\begin{equation}
    f(v) =
\begin{cases}
0, & v<0,\\
\lambda v (\lambda >0) & 0 \leq v\leq T,\\ f_0(>0), & v>T,
\end{cases}
\end{equation}
and
\begin{equation}
     e^2(v) =
\begin{cases}
 0, & v<0,\\ \mu^2v^2(\mu^2>0), & 0 \leq v\leq T,\\ e^2_0(>0), & v>T.
\end{cases}
\end{equation}
\end{subequations}
For this choice of mass function, using equation \eqref{limitingvalueX02} we obtain
\begin{equation}\label{X_0chargedvaidya}
    \mu^2X_0^3 - 2\lambda X_0^2 + X_0 -2 = 0.
\end{equation}
This equation is a polynomial of degree three with the negative last term and positive first coefficient. By the theory of polynomial functions, every equation of this nature must have at least one root which is positive. The existence of these roots signifies that the singularity is naked. In particular, when $\mu^2 = 0.001, \; \lambda = 0.01$, then one of the roots of equation \eqref{X_0chargedvaidya} is $2.077$ and $ \lim\limits_{k\to 0}k^2\psi = \frac{1}{2} X_{0}^{2}\left(\lambda - \mu^2X_0\right) = 0.0171 > 0$. Therefore the condition for a strong naked singularity is satisfied.

\item {\it The charged Vaidya-deSitter spacetime}

The charged Vaidya-deSitter solution is a generalised Vaidya solution of a charged null fluid in an expanding de-Sitter background \cite{Beesham}.
We define the mass mass function as
\begin{equation}\label{Antidesitter}
    m(v,r) = m(v) - \frac{e^2(v)}{2r} + \frac{\Lambda r^3}{6},
\end{equation}
where $f(v)$ and $e(v)$ are arbitrary functions representing the mass and electric charge respectively, and
$\Lambda \neq 0$ is the cosmological constant. For the weak energy condition to be satisfied, it is required that $r\dot{m}(v) - e(v)\dot{e}(v)$ to be non-negative \cite{Wang,Beesham}. We specifically define the functions similar to that of charged Vaidya and the algebraic equation that governs the behaviour of the tangent vectors near the central singularity comes out to be the same.

 \item {\it The Husain solution}

This is a solution of the Einstein field equations for the null fluid with the equation of state $\varrho = k\rho$ where $\rho = \frac{g(v)}{4\pi r^{2k+2}},\: k \neq \frac{1}{2}$ \cite{Husian_1996, Wang}. This solution is a subclass to the generalised Vaidya solutions with the mass function given by
\begin{subequations}
\begin{equation}\label{Husinasoln}
    m(v,r) =
    \begin{cases}
    q(v) - \frac{g(v)}{(2k - 1)r^{2k-1}}, & k\neq \frac{1}{2},\\
    q(v) + g(v)\ln r, & k = \frac{1}{2},
    \end{cases}
\end{equation}
where $q(v)$ and $g(v)$ are arbitrary functions which are restricted only by the energy conditions.
For the dominant energy conditions to be satisfied, it is required that $g(v)\geq 0$ and either $\dot{g}(v) > 0$ for $k < \frac{1}{2}$ or $\dot{g}(v) < 0$ for $k > \frac{1}{2}$. The weak or strong energy conditions  are satisfied when $\rho \geq 0,\: \varrho \geq 0$. We consider the case when $k \neq \frac{1}{2}$ and define the mass function as
\begin{equation}
    q(v) =
\begin{cases}
0, & v<0,\\ \frac{1}{2}\lambda v (\lambda >0), & 0 \leq v\leq T,\\ q_0(>0), & v>T,
\end{cases}
\end{equation} and
\begin{equation}
g(v) =
\begin{cases}
0, & v<0,\\ \mu^2v^{2k}, & 0 \leq v\leq T,\\ g_0(>0), & v>T.
\end{cases}
\end{equation}
\end{subequations}
For this mass function using equation \eqref{limitingvalueX02}, we get
\begin{equation}\label{X_0Husian}
    2\mu^2\left(1 - \frac{2k}{2k-1}\right)X_0^{2k+1} + \lambda X_0^2 - 2X_0 + 4 = 0.
\end{equation}
This equation can be solved to get some positive roots $X_0$ for some particular values of $\mu^2, \; k$ and $\lambda$. In particular, when $\mu^2 = 0.001, \; k = \lambda = 0.01$, then one of the roots of equation \eqref{X_0Husian} is $2.00408$ and $\lim\limits_{k\to 0}k^2\psi = \frac{1}{4}X_{0}^{2}\left(\lambda -\frac{4k\mu^2}{2k-1}X_0^{2k-1}\right)= 0.506 > 0$. This shows that the singularity is naked and strong.
\end{enumerate}
Table \ref{Table_X_0} gives a summary of the equations of tangent to the singularity curve $X_0$ and the value of $\lim\limits_{k\to 0}k^2\psi $ for chosen mass functions in some sub-classes of the generalised Vaidya spacetime.

\section{Concluding Remarks\label{conclusion}}
In this paper we developed a general mathematical formalism to study the gravitational collapse of the generalised Vaidya spacetime in the context of the Cosmic Censorship Conjecture. We studied the structure of the central singularity to show that it can be a node with outgoing radial null geodesics emerging from the singular point with definite value of the tangent, 
depending on the nature of the generalised Vaidya mass function and the parameters in the problem. The key points that emerged transparently from this analysis are as follows:
\begin{itemize}
\item It is quite clear that given any realistic mass function, there always exists an open set in the parameter space for which the central singularity is naked and CCC is violated. A similar result is well known for pure Type I matter fields \cite{Joshibook2, goswami}. Hence we can conclude that the occurrence of naked singularity is quite a "stable" phenomenon even when the nature of matter field changes by combining a radiation-like field along with a collapsing perfect fluid.
\item It is also evident that for an open set in the parameter space, these naked central singularities are strong and they cannot be regularised anyway by extension of spacetime through them. This has far reaching consequences as their presence will no longer make the global spacetime future asymptotically simple, and the proofs of black hole dynamics and thermodynamics have to be reformulated.
\item Finally the generalised Vaidya spacetime is a more realistic spacetime than pure dust-like matter or perfect fluid, during the later stages of gravitational collapse of a massive star. A collapsing star should always radiate and hence there should be a combination of light-like matter along with a perfect fluid. Therefore a violation of censorship in these models should have novel astrophysical signatures which are yet to be properly deciphered.

\end {itemize}

\section {Acknowledgement \label{acknowledge}}
We are indebted to the National Research Foundation and the University of KwaZulu-Natal for financial support.
SDM acknowledges that this work is based upon research supported by the South African Research Chair Initiative of the
Department of Science and Technology. MDM extends his appreciation to the University of Dodoma in Tanzania for study leave.

\end{document}